\documentclass[pdflatex,sn-mathphys-num]{sn-jnl}

\usepackage{amsmath}
\usepackage{amssymb}
\usepackage{graphicx}
\usepackage{amsthm}
\usepackage{subcaption}
\usepackage{a4wide}

\newtheorem{definition}{Definition}
\newtheorem{theorem}{Theorem}
\newtheorem{lemma}{Lemma}
\newtheorem{corollary}{Corollary}
\newtheorem{proposition}{Proposition}
\newtheorem{example}{Example}
\theoremstyle{remark}
\newtheorem{remark}{Remark}

\begin{document}

\title{A consistent SIR model on time scales with exact solution}

\author[1,2]{\fnm{M\'{a}rcia} \sur{Lemos-Silva}}
\email{marcialemos@ua.pt}
% {https://orcid.org/0000-0001-5466-0504}

\author[3]{\fnm{Sandra} \sur{Vaz}}
\email{svaz@ubi.pt}
% {https://orcid.org/0000-0002-1507-2272}

\author*[1,2]{\fnm{Delfim F. M.} \sur{Torres}}
\email{delfim@ua.pt}
% {https://orcid.org/0000-0001-8641-2505}

% -------------------------------------------------------

\affil[1]{Center for Research and Development in Mathematics and Applications (CIDMA),\newline 
Department of Mathematics,
University of Aveiro,
3810-193 Aveiro, Portugal}

\affil[2]{Research Center in Exact Sciences (CICE), 
Faculty of Sciences and Technology,\newline
University of Cape Verde (Uni-CV),
7943-010 Praia, Cabo Verde}

\affil[3]{Center of Mathematics and Applications (CMA-UBI),
Department of Mathematics,
University of Beira Interior,
6201-001 Covilh\~{a}, Portugal}

% -------------------------------------------------------

\abstract{We propose a new dynamic SIR model that, 
in contrast with the available model on time scales, 
is biological relevant. For the new SIR model we obtain
an explicit solution, we prove the 
asymptotic stability of the extinction 
and disease-free equilibria, and deduce some
necessary conditions for the monotonic behavior of the infected population.
The new results are illustrated with several examples in the discrete, continuous,
and quantum settings.}

\keywords{Compartmental models, 
Nonlinear dynamic SIR systems, 
Exact solution,
Time scales.}

\pacs[MSC Classification]{34N05, 92D30.}

\maketitle

% ---------------------------------------

[\noindent
This is a preprint version of the paper published open access 
in \emph{Nonlinear Dynamics} at\\
\url{https://doi.org/10.1007/s11071-025-11758-0}]

% ---------------------------------------

\section{Introduction}
\label{sec1}

The modeling of infectious diseases has been a crucial component of public health 
strategy since 1927, when Kermack and McKendrick introduced the celebrated SIR 
model \cite{kermack:mckendrick}. Such model is based on the idea of dividing
the population into three different groups: susceptible, infected and removed. 
Nearly a century later, the importance of disease modeling persists, as it 
remains a valuable tool for predicting outbreaks, designing interventions 
and policy decisions \cite{MyID:468,TM1,Hattaf2013}.

In this paper, we consider a version of the susceptible-infected-removed model 
proposed by Bailey in \cite{bailey}. This model considers interactions between 
the three groups, governed by the disease's transmission and removal rates. 
Moreover, the removed compartment is introduced solely to ensure the population 
remains constant over time. Thus, susceptible individuals interact not
with the total population but with the combined sum of the susceptible and 
infected groups. Mathematically, such SIR model takes the form 
\begin{equation}
\label{sir:continuous}
\begin{cases}
x' = -\dfrac{b x y}{x+y}, \\
y' = \dfrac{b x y}{x+y} - cy, \\
z' = cy,
\end{cases}
\end{equation}
where $b, c \in \mathbb{R}^+_0$, $x(t_0) = x_0 > 0$, $y(t_0) = y_0 > 0$ 
and $z(t_0) = z_0 \geq 0$ for some $t_0 \in \mathbb{R}_0^+$. Moreover, 
$x, y, z : \mathbb{R} \rightarrow \mathbb{R}^+_0$. Here, $x$ denotes the
susceptible population, $y$ the infected population, and finally $z$ the
removed population. Observe that from \eqref{sir:continuous}, we find that 
$x'(t) + y'(t) + z'(t) = 0$. This implies that $x(t) + y(t) + z(t) = N$, 
where $N := x_0 + y_0 + z_0$ represents the constant population being 
analyzed. 

The SIR model is very simple, but it forms the foundation of all compartmental 
modeling in epidemiology \cite{MR4368314,s11071-024-09710-9,s11071-025-11006-5}. 
Despite its apparent simplicity, the model is nonlinear, and 
obtaining an analytical solution is far from trivial. In fact, most textbooks and papers 
either do not mention or are unaware that an exact solution in closed form exists under 
certain assumptions. This underlines the mathematical challenges involved and further 
motivates continued research on analytical approaches for such ``simple'' nonlinear systems
\cite{sir:ts,MR4492457,MR4692913}.

Over the years, the field of mathematical epidemiology has seen remarkable advancements, leading 
to more complex and realistic frameworks that allow for more accurate and 
robust predictions (see, e.g., \cite{caetano, gai, lopez, influenza}). 
Nonetheless, to the best of our knowledge, limited progress has been made
in modeling infectious diseases on a general time scale. 
Here we propose and investigate a SIR model within the framework of time scales, 
constructed in a mathematically consistent way. The theory of time scales offers 
a powerful tool for bridging discrete and continuous analysis. This is particularly relevant 
in the context of biological and epidemiological systems, where the timing of events, 
such as infections, recoveries, or interventions, may not follow a strictly continuous 
or discrete pattern. With this, the proposed model has the ability to reflect the mixed 
and often irregular nature of real-world data, enabling a more faithful representation 
of disease dynamics across a wide range of scenarios. 

In \cite{sir:ts}, a dynamic version of model \eqref{sir:continuous} is proposed with 
time-dependent parameters. Unfortunately, as shown in \cite{sir:discrete}, 
the dynamic model of \cite{sir:ts} fails to guarantee consistency in 
the discrete time domain, i.e., when $\mathbb{T} = h\mathbb{Z}$. 
Specifically, it has been proved that, even with non-negative initial 
conditions, the model of \cite{sir:ts} can yield negative results, which, 
while mathematically correct, lack biological relevance. 
The same problem occurs in the recent work \cite{MyID:561}.
In this paper, our main goal is to present a dynamic analogue 
of system~\eqref{sir:continuous}, 
also with time-dependent coefficients, that, 
unlike \cite{MyID:561,sir:ts}, ensures 
consistency for any arbitrary time scale. 
Furthermore, unlike many SIR models, 
our proposed model has an exact solution. 

This paper is organized as follows. In Section~\ref{sec2}, some fundamental
results of the theory of time scales are introduced. Our new, consistent,
and meaningful dynamic SIR model is then presented in Section~\ref{sec3}, 
along with its exact solution (Theorem~\ref{thm:6}). We proceed with the 
analysis of the asymptotic behavior of our model, illustrating the main results 
(Theorem~\ref{thm:7} and Theorem~\ref{thm:8}) with some applications 
in important time scales. Our last example analyzes
the particular dynamics of the infected population with the value
of the basic reproduction number in agreement with Theorem~\ref{thm:9}. 
We end with Section~\ref{sec:conc} of conclusions and future work.

%------------------------------------------------------

\section{Time-scale fundamentals}
\label{sec2}

In this section, we introduce some fundamental concepts of time scales 
that are crucial for the sequel. For more details,
we refer the interested reader to the books 
\cite{book:ts, advanced:book:ts}.

A time scale $\mathbb{T}$ is any nonempty closed subset of the real numbers.

\begin{definition}[See \cite{book:ts}]
For $t \in \mathbb{T}$, the forward jump operator $\sigma : \mathbb{T} 
\rightarrow \mathbb{T}$ is
$$\sigma(t) := \inf\{s \in \mathbb{T} : s > t\},$$
while the backward jump operator $\rho : \mathbb{T} \rightarrow \mathbb{T}$
is $$\rho(t) := \sup\{s \in \mathbb{T} : s < t\}.$$
\end{definition}  

\begin{definition}[See \cite{book:ts}]
The graininess function $\mu(t): \mathbb{T} \rightarrow [0, \infty)$ is 
defined as
$$\mu(t) := \sigma(t) - t.$$ 
\end{definition}

\begin{definition}[See \cite{book:ts}]
A function $f: \mathbb{T} \rightarrow \mathbb{R}$ is called rd-continuous 
provided it is continuous at right-dense points in $\mathbb{T}$ and its 
left-sided limits exists for all left-dense points in $\mathbb{T}$. The set 
of rd-continuous functions $f: \mathbb{T} \rightarrow \mathbb{R}$ is denoted 
by $C_{rd} = C_{rd}(\mathbb{T}) = C_{rd}(\mathbb{T},\mathbb{R})$. 
\end{definition}

\begin{definition}[See \cite{book:ts}]
A function $f: \mathbb{T} \rightarrow \mathbb{R}$ is called regressive provided
$$1 + \mu(t)f(t) \neq 0\, \text{for all}\,\,\, t \in \mathbb{T}.$$
The set of regressive and rd-continuous functions is denoted by 
$$\mathcal{R} = \mathcal{R}(\mathbb{T}) = \mathcal{R}(\mathbb{T}, \mathbb{R}).$$
Moreover, $f \in \mathcal{R}$ is called positively regressive,
i.e., $f \in \mathcal{R}^+$, if 
$$1 + \mu(t)f(t) > 0\, \ \text{for all}\,\,\, t \in \mathbb{T}.$$ 
\end{definition}

If $t \in \mathbb{T}$ has a left-scattered maximum $m$, then 
$\mathbb{T}^\kappa = \mathbb{T} \setminus \{m\}$; otherwise 
$\mathbb{T}^\kappa = \mathbb{T}$.

\begin{definition}[See \cite{book:ts}]
\label{delta:derv}
Let $f: \mathbb{T} \rightarrow \mathbb{R}$ be a function and 
$t \in \mathbb{T}^\kappa$. Then $f^\Delta(t)$ denotes the delta (or Hilger) 
derivative and we define it as the number (provided it exists) for which given 
any $\varepsilon > 0$ there is a neighborhood $U$ of $t$, 
$U = (t - \delta, t + \delta) \cap \mathbb{T}$ for some $\delta > 0$, such that 
$$\lvert f(\sigma(t)) - f(s) - f^\Delta(t)[\sigma(t) - s]\rvert 
\leq \varepsilon \lvert \sigma(t) - s \rvert,$$
for all $s \in U$. Moreover, if $f^\Delta(t)$ exists for all 
$t \in \mathbb{T}^\kappa$, we say that $f$ is delta differentiable. 
\end{definition}

\begin{theorem}[See \cite{book:ts}]
\label{thm:1}
Let $p \in \mathcal{R}$ and $t_0 \in \mathbb{T}$. Then the regressive IVP problem 
of the form
$$y^\Delta = p(t)y,\quad y(t_0) = 1,$$
has the exponential function as its unique solution, denoted by $e_p(\cdot, t_0)$. 
\end{theorem}

\begin{theorem}[See \cite{book:ts}]
\label{thm:2}
Assume $t_0 \in \mathbb{T}$. If $p(t) \in \mathcal{R}^+$, then $e_p(t,t_0) > 0$ 
for all $t \in \mathbb{T}$.
\end{theorem}

Now we recall the properties of the exponential function that are relevant
for the subsequent analysis. More known properties of the exponential function 
can be found in \cite{book:ts}.

\begin{theorem}[See \cite{book:ts}]
\label{thm:3}
If $p,\,q \in \mathcal{R}$, then
\begin{itemize}
\item $e_0(t,s) \equiv 1$ and $e_p(t,t) \equiv 1$;
\item $e_p(t,s) = \frac{1}{e_p(s,t)} = e_{\ominus p}(s,t)$;
\item $e_p(t,s)e_q(t,s) = e_{p \oplus q}(t,s)$.
\end{itemize}
\end{theorem}

\begin{theorem}[See \cite{book:ts}]
\label{thm:4}
If $p \in \mathcal{R}$ and $a,\,b,\,c \in \mathbb{T}$, then 
$$\int_{a}^{b} p(t)e_p(c,\sigma(t))\Delta(t) = e_p(c,a) - e_p(c,b),$$
and
$$\int_{a}^{b} p(t)e_p(t,c) \Delta(t) = e_p(b,c) - e_p(a,c).$$
\end{theorem}

\begin{theorem}[Variation of constants \cite{advanced:book:ts}]
\label{thm:5}
Suppose $p \in \mathcal{R}$ and $f \in C_{rd}$. Let $t_0 \in \mathbb{T}$ 
and $y_0 \in\mathbb{R}$. The unique solution of the IVP 
$$y^\Delta = p(t)y + f(t),\quad y(t_0) = t_0,$$
is given by
\begin{equation*}
\begin{aligned}
y(t) &= e_p(t,t_0)y_0 + \int_{t_0}^{t} e_p(t, \sigma(\tau))f(\tau)\Delta\tau.
\end{aligned}
\end{equation*}
Moreover, the unique solution of the IVP
$$y^\Delta = -p(t)y^\sigma + f(t),\quad y(t_0) = t_0,$$
is given by 
\begin{equation*}
\begin{aligned}
y(t) &= e_{\ominus p}(t,t_0)y_0 + \int_{t_0}^{t} 
e_{\ominus p}(t, \tau)f(\tau)\Delta\tau.
\end{aligned}
\end{equation*}
\end{theorem}

The next two lemmas will play a crucial role in analyzing the long-term 
behavior of solutions. Both can be found in more recent works \cite{osc, prop}.

\begin{lemma}[See \cite{prop}]
\label{lemma:1}
If $p \in \mathcal{R}^+$, then
$$
0 < e_p(t,t_0) \leq \exp\left\{\int_{t_0}^{t} p(\tau) \,\Delta\tau\right\},
\quad \text{for all}\,\, t\geq t_0. 
$$
\end{lemma}

\begin{lemma}[See \cite{osc}]
\label{lemma:2}
If $p \in C_{rd}$ and $p(t) \geq 0$ for all $t \in \mathbb{T}$, then
$$
1 + \int_{t_0}^t p(\tau) \,\Delta\tau 
\leq e_p(t,t_0) \leq \exp\left\{\int_{t_0}^t p(\tau)\, \Delta\tau\right\},
\quad \text{for all}\,\, t\geq t_0. 
$$ 
\end{lemma}

%------------------------------------------------------

\section{Main Results}
\label{sec3}

There are several possible ways to discretize system \eqref{sir:continuous} 
in order to obtain a discrete-time version. The choice of discretization method can 
significantly affect the qualitative behavior of the resulting model. For instance, 
applying the explicit Euler method to system \eqref{sir:continuous} leads to
\begin{equation}
\label{discExp}	
\begin{cases}
x_{n+1} = x_n - h \dfrac{b x_n y_n}{x_n + y_n}, \\[10pt]
y_{n+1} = y_n + h \left( \dfrac{b x_n y_n}{x_n + y_n} - c y_n \right), \\[10pt]
z_{n+1} = z_n + h \cdot c y_n,
\end{cases}
\end{equation}
where $ h> 0$ is the time step. On the other hand, 
using the implicit Euler method gives
\begin{equation}
\label{discImp}	
\begin{cases}
x_{n+1} = x_n - h \dfrac{b x_{n+1} y_{n+1}}{x_{n+1} + y_{n+1}}, \\[10pt]
y_{n+1} = y_n + h \left( \dfrac{b x_{n+1} y_{n+1}}{x_{n+1} + y_{n+1}} - c y_{n+1} \right), \\[10pt]
z_{n+1} = z_n + h \cdot c y_{n+1}.
\end{cases}
\end{equation}
Each method has its own advantages and drawbacks. For example, the explicit Euler method 
is simpler to implement but may lead to instability or loss of qualitative properties such 
as non-negativity \cite{Stuart96}. Even though the implicit Euler method is typically more
stable, it does not ensure the preservation of non-negativity. In fact, by solving 
the equation \eqref{discImp} for $y_{n+1}$, we obtain:
$$
y_{n+1} = \frac{y_n}{1 - h \left( \dfrac{b x_{n+1}}{x_{n+1} + y_{n+1}} - c \right)}.
$$
This expression shows that $y_{n+1}$ can become negative even if $y_n > 0$, depending on 
the parameters and the step size $h$, since the denominator can be negative. Therefore, 
the implicit Euler method does not guarantee the non-negativity of the solution without 
further conditions. Thus, here we aim to formulate a time-scale model in such a way that all 
the essential structural properties of the original model \eqref{sir:continuous} 
are preserved across different time scales. 
This ensures consistency with both the continuous and discrete cases, 
while respecting the specific features of the time scale calculus. Precisely, we propose 
the following dynamic SIR model on time scales:
\begin{equation}
\label{dynamic:sir}
\begin{cases}
x^\Delta(t) = -\dfrac{b(t)x^\sigma(t)y(t)}{x(t)+y(t)},\\
y^\Delta(t) = \dfrac{b(t)x^\sigma(t)y(t)}{x(t)+y(t)} - c(t)y^\sigma(t),\\
z^\Delta(t) = c(t)y^\sigma(t),
\end{cases}	
\end{equation}
where $b, c : \mathbb{T} \rightarrow \mathbb{R}_0^+$, $x(t_0),\, y(t_0) > 0$ 
and $z(t_0) \geq 0$. Moreover, $x, y: \mathbb{T} \rightarrow \mathbb{R}^+$ and 
$z:~ \mathbb{T} \rightarrow \mathbb{R}^+_0$. 

\begin{remark}
\label{rem:constant:t}
Note that 
$x^\Delta(t) + y^\Delta(t) + z^\Delta(t) = 0$, which means that the 
population remains constant over time, and thus $z(t) = N - x(t) - y(t)$, 
for all $t \in \mathbb{T}$. 
\end{remark}

In continuous time (i.e., when $\mathbb{T}=\mathbb{R}$) the forward jump operator is the 
identity operator ($\sigma(t) = t$) and its position and relevance in a differential equation 
is not visible. However, on other time scales (e.g., $\mathbb{T}=h\mathbb{Z}$, 
$\mathbb{T}=q^{\mathbb{N}_0}$, etc.) its presence is crucial 
and its position changes completely the dynamics of the system. 
For instance, in the language of time scales, 
system \eqref{discExp} takes the form
\begin{equation}
\label{discExp2}	
\begin{cases}
x^\sigma(t) = x(t) - \mu \dfrac{b x(t) y(t)}{x(t) + y(t)}, \\[10pt]
y^\sigma(t) = y(t) + \mu  \left( \dfrac{b x(t) y(t)}{x(t) + y(t)} - c y(t) \right), \\[10pt]
z^\sigma(t) = z(t) + \mu  \cdot c y(t),
\end{cases}
\end{equation}
while \eqref{discImp} is equivalent to	
\begin{equation}
\label{discImp2}	
\begin{cases}
x^\sigma(t) = x(t) - \mu \dfrac{b x^\sigma(t) y^\sigma(t)}{x^\sigma(t) + y^\sigma(t)}, \\[10pt]
y^\sigma(t) = y(t) + \mu \left( \dfrac{b x^\sigma(t) y^\sigma(t)}{x^\sigma(t) + y^\sigma(t)} - c y^\sigma(t) \right), \\[10pt]
z^\sigma(t) = z(t) + \mu \cdot c y^\sigma(t),
\end{cases}
\end{equation}
where $t \in h\mathbb{Z}$ and $\mu(t) \equiv h =: \mu$. 
We see that for $\mathbb{T} = h\mathbb{Z}$ the position 
of the $\sigma$ operator plays a crucial role in the formulation of discrete dynamical systems
and, similarly, the placement of $\sigma$ is equally important and 
must be handled with care in a general time scale. As already discussed, 
classical numerical methods may fail to preserve essential
structural properties of the original model, such as the non-negativity of solutions. 
To avoid such inconsistencies and potential numerical instabilities, one needs to consider 
a non-classical finite difference scheme. Mickens' nonstandard finite difference scheme has 
been widely used to the discretization of dynamical systems due to its ability to preserve key 
qualitative properties of the original model \cite{sir:discrete,MR4276109}. 

According to \cite{Mickens05}, a dynamically consistent non-standard finite 
difference scheme depends strongly on a rule that states that both linear 
and nonlinear terms of the state variables and their derivatives 
may need to be substituted by nonlocal forms. 
Having this in mind, in our model \eqref{dynamic:sir}
we adapted the underlying idea of the Mickens' method to the broader 
context of time scales calculus, ensuring, as we will prove, 
that the resulting formulation preserves the
qualitative properties of the original model across different time domains. 
In particular, as already discussed, 
applying non-local forms uniformly to all variables 
does not ensure the preservation of non-negativity 
in the solutions. Therefore, the placement of the $\sigma$ operator must be done 
strategically to guarantee this critical property. Our next results
justify the well-posedness of the proposed model \eqref{dynamic:sir}.

\begin{theorem}[Explicit solution to the dynamic SIR model \eqref{dynamic:sir}]
\label{thm:6}
If $f,\, g \in \mathcal{R}$, then the unique solution of 
\eqref{dynamic:sir} is given by
\begin{equation}
\label{dynamic:solution}
\begin{cases}
x(t) = e_{\ominus g} (t,t_0)\,x_0,\\
y(t) = e_{\ominus \left(f \oplus g\right)}(t,t_0)\,y_0,\\
z(t) = N - e_{\ominus g}(t,t_0)\left[x_0 + e_{\ominus f}(t,t_0)\,y_0\right],
\end{cases}
\end{equation}
where 
\begin{equation}
\label{eq:f:g}
f(t):= \frac{(c-b)(t)}{1 + b(t)\mu(t)},\,\quad g(t) 
:= \frac{b(t)\kappa}{e_f(t,t_0) + \kappa},
\end{equation}
$\kappa = \frac{y_0}{x_0}$ and $N = x_0 + y_0 + z_0$.
\end{theorem}

\begin{proof}
Let us define $\omega = \frac{x}{y}$. Then, we have
$$
\omega^\Delta = \frac{x^\Delta y - y^\Delta x}{yy^\sigma}
= \frac{\frac{-bx^\sigma y}{x + y}y - \left(\frac{bx^\sigma y}{x+y} 
- cy^\sigma\right)x}{yy^\sigma}
= \frac{-bx^\sigma y + cy^\sigma x}{yy^\sigma}
= -b\omega^\sigma + c\omega.
$$
Since $\omega$ is differentiable at $t$, then 
$\omega^\sigma = \omega + \mu(t) \omega^\Delta$. Thus, 
$$
\omega^\Delta = -b(t)(\omega + \mu(t)w^\Delta) + c(t)\omega 
\Leftrightarrow \omega^\Delta = \frac{(c-b)(t)}{1 + b(t)\mu(t)}\omega
\Leftrightarrow \omega^\Delta = f(t) \omega, 
$$
which is a first-order linear dynamic equation whose solution is known
and is given by
$$
\omega(t) = e_{f}(t,t_0)\,\omega_0. 
$$
So, 
\begin{equation}
\label{y:sol}
y(t) = \kappa\, e_{\ominus f}(t,t_0)\,x(t),
\end{equation}
with $\kappa = \frac{y_0}{x_0}$. 
Plugging \eqref{y:sol} into \eqref{dynamic:sir}, we get
$$
x^\Delta = -\frac{b(t)\kappa e_{\ominus f}(t,t_0)}
{1 + \kappa e_{\ominus f}(t,t_0)}\,x^\sigma,
$$
which has the solution
$$
x(t) = e_{\ominus g}(t,t_0)\,x_0,
$$
where 
$$
g(t)= \frac{b(t)\kappa e_{\ominus f}(t,t_0)}
{1 + \kappa e_{\ominus f}(t,t_0)}
= \frac{b(t)\kappa e_{\ominus f}(t,t_0)}{1 + \kappa e_{\ominus f}(t,t_0)}
= \frac{b(t)\kappa}{e_f(t,t_0) + \kappa}.
$$
By \eqref{y:sol}, we have 
$$
y(t) = \kappa e_{\ominus f}(t,t_0)\cdot
e_{\ominus g}(t,t_0)\, x_0 = e_{\ominus(f \oplus g)}(t,t_0)\,y_0.
$$
Finally, as $z(t) = N - x(t) - y(t)$, we obtain that
$$
z(t) = N - e_{\ominus g}(t,t_0)\left[x_0 + e_{\ominus f}(t,t_0)\,y_0\right].
$$
The proof is complete. 
\end{proof}

Theorem~\ref{thm:6} provides an explicit solution to the SIR model formulated on an arbitrary 
time scale. This result is particularly important as it unifies, within a single framework, 
the solutions corresponding to both the continuous-time and discrete-time SIR models. 
In other words, by choosing different time scales (e.g., $\mathbb{T} = \mathbb{R}$ for 
continuous time, or $\mathbb{T} = h\mathbb{Z}$ for discrete time), the general solution 
obtained recovers the classical solutions known in the literature: see
\cite{cont:sol} and \cite{sir:discrete}, respectively for the continuous 
and discrete-time cases.
In the classical continuous case, the exact solution involves exponential functions. 
Our result \eqref{dynamic:solution} -- 
a single and elegant expression valid across a wide variety of time domains -- 
provides a natural generalization of the classical solutions through the use 
of time scale exponentials. Moreover, it preserves essential qualitative properties 
such as conservation of total population (Remark~\ref{rem:constant:t})
and, as we shall prove, non-negativity of solutions (Proposition~\ref{prop:nonneg}), 
while providing a unified 
expression that recovers continuous and discrete cases as particular examples. This makes our 
model \eqref{dynamic:sir} highly relevant for applications where the disease dynamics 
cannot be fully captured by purely continuous or purely discrete models. Before
proving it, we compare, in terms of graphs, the solution of our model 
\eqref{dynamic:sir} with those available in the literature.

For $\mathbb{T} = \mathbb{R}$ one obtains from \eqref{dynamic:sir}
the classical continuous-time SIR system \eqref{sir:continuous}.
The exact solution of model \eqref{sir:continuous} 
was first derived in \cite{cont:sol} for the case when $b$
and $c$ are constants. Figure~\ref{infected_constant} illustrates 
the dynamics of the infected population as 
given by the exact solution of system \eqref{sir:continuous}, as found
in \cite{cont:sol}, our proposed model \eqref{dynamic:sir} with $\mathbb{T}=\mathbb{Z}$, 
and the existing time-scale model \cite{sir:ts} from the literature with $\mathbb{T}=\mathbb{Z}$, 
considering constant parameters. The difference between the solution of model \cite{sir:ts}, 
our model \eqref{dynamic:sir} and the continuous-time system \eqref{sir:continuous} 
is striking: the time-scale model \cite{sir:ts} allows for negative solutions, 
which jeopardizes its biological relevance. 
In contrast, our model preserves non-negativity, ensuring consistency with the fundamental 
principles of the classical model \eqref{sir:continuous}.
% -----------------------------
\begin{figure}[ht]
\centering
\includegraphics[scale=0.6]{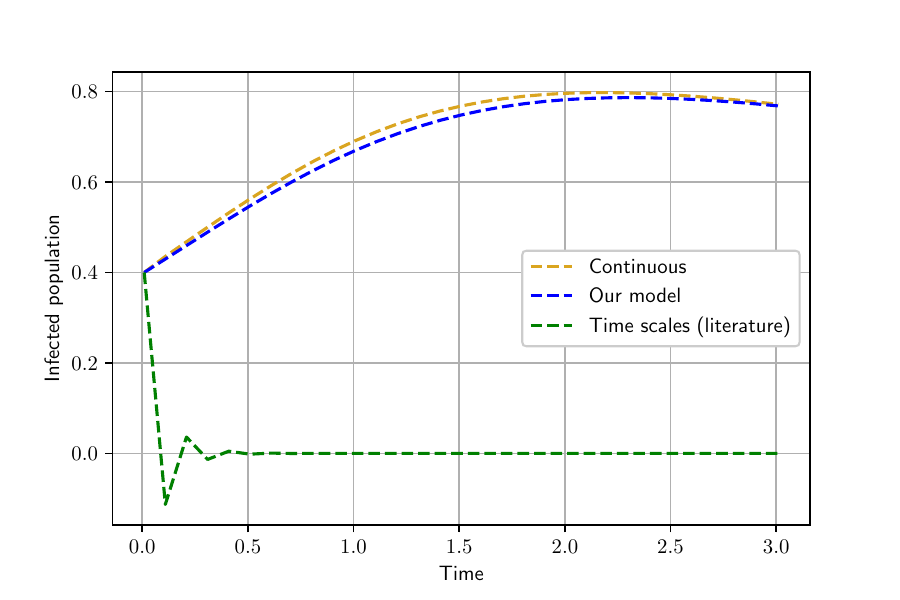}
\caption{Dynamics of infected population 
with constant coefficients $b = 1.5$ and $c = 0.1$	
for the time-scale models \eqref{dynamic:sir} and \cite{sir:ts} 
with $\mathbb{T} = \mathbb{Z}$ versus the classical 
continuous SIR model \eqref{sir:continuous}.}
\label{infected_constant} 
\end{figure}
% -----------------------------

Figure~\ref{infected_variable} presents a comparison of the dynamics of the infected 
population, under time-varying parameters. More precisely, we consider
\begin{equation}
\label{eq:bt}
b(t) = 0.8 - 0.6\sin(t), 
\end{equation}
and
\begin{equation}
\label{eq:ct}
c(t) = \frac{t}{1+t}.
\end{equation}
Once again, the time-scale model from \cite{sir:ts} differs significantly from both our 
model \eqref{dynamic:sir} and the classical continuous SIR model 
\eqref{sir:continuous}, highlighting its inconsistency when parameter variability is 
introduced.
% -----------------------------
\begin{figure}[ht]
\centering
\includegraphics[scale=0.6]{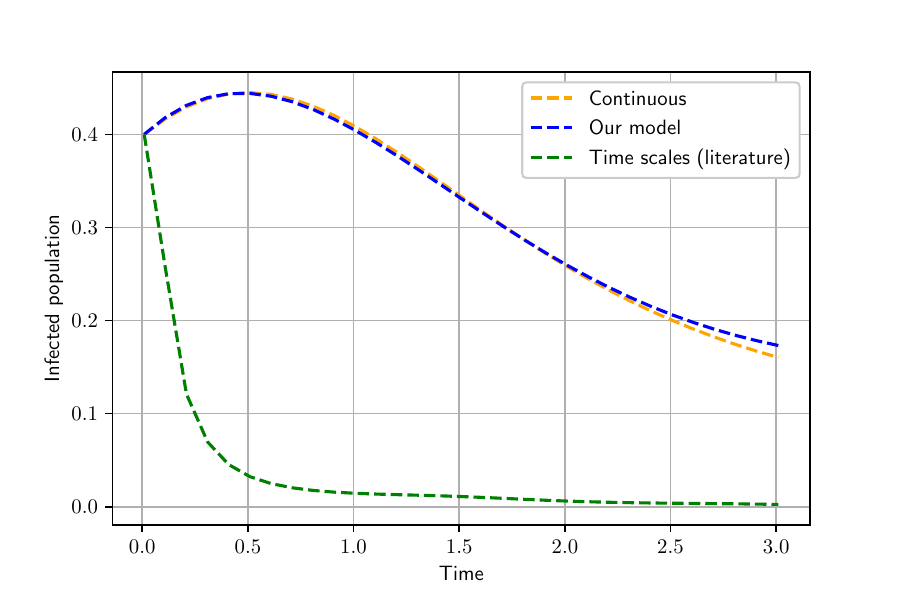}
\caption{Dynamics of infected population 
with time-dependent coefficients \eqref{eq:bt} and \eqref{eq:ct}	
for the time-scale models \eqref{dynamic:sir} and \cite{sir:ts} 
with $\mathbb{T} = \mathbb{Z}$ versus the classical 
continuous SIR model \eqref{sir:continuous}.}
\label{infected_variable} 
\end{figure}
% -----------------------------

\begin{lemma}
\label{lemma:3}
Consider $f$ and $g$ as defined in \eqref{dynamic:solution}--\eqref{eq:f:g}.
Then, $f,\,g \in \mathcal{R}^+$. 
\end{lemma}

\begin{proof}
Let $f(t)$ and $g(t)$ be as defined in Theorem \ref{thm:6}. From $f(t)$ we 
conclude that
$$
1 + \mu(t) f(t) > 0 \Leftrightarrow c(t) > -\frac{1}{\mu(t)},
$$
which is always true since $c: \mathbb{T} \rightarrow \mathbb{R}^+_0$. 
Thus, $f(t) \in \mathcal{R}^+$.  
Moreover, if $b(t) = 0$, then $g(t) = 0$, and so $1 + \mu(t)g(t) = 1 > 0$.
On the other hand, if $b(t) > 0$, then we have $g(t) > 0$ for all $t \in \mathbb{T}$. 
In both cases, it is clear that $g(t) \in \mathcal{R}^+$.
Note that for $\mathbb{T} = \mathbb{R}$, conditions for both $f$ and $g$ to be 
positively regressive are trivially satisfied since $\mu(t) \equiv 0$, 
for all $t \in \mathbb{R}$. The proof is complete.
\end{proof}

\begin{proposition}[Biological relevance of the solution 
to the dynamic SIR model \eqref{dynamic:sir}]
\label{prop:nonneg}
Consider \eqref{dynamic:solution}. If $x_0,\,y_0 > 0$ and $z_0 \geq 0$, 
then $x,\,y,\,z$ are non-negative for all $t \in \mathbb{T}$.  
\end{proposition}

\begin{proof}
Follows trivially from Theorem~\ref{thm:2} and Lemma~\ref{lemma:3}.
\end{proof}

\begin{lemma}[Equilibria of the dynamic SIR model \eqref{dynamic:sir}]
\label{lemma:4}
Suppose $c(t) > 0$ at some point $t \in \mathbb{T}$. The equilibria of 
\eqref{dynamic:sir} are $(\alpha, 0, N-\alpha)$, where $\alpha \in [0,N]$ 
and $N = x_0 + y_0 + z_0$. 
\end{lemma}

\begin{proof}
Let $x,\,y,\,z$ be solutions of \eqref{dynamic:sir}. If $c(t) > 0$, then
$$
z^\Delta(t) = c(t)y(t) = 0 \Leftrightarrow y(t) = 0. 
$$
In this case, 
$$
x^\Delta = -\frac{b(t)x(t)y(t)}{x(t) + y(t)} = 0, 
$$
for all $x \in [0,N]$, and 
$$
y^\Delta = \frac{b(t)x(t)y(t)}{x(t) + y(t)} - c(t)y(t) = 0.
$$
Thus, the result holds. 
\end{proof}

\begin{theorem}[Asymptotic stability of the extinction equilibrium]
\label{thm:7}
Consider \eqref{dynamic:sir} and let $\mathbb{T}$ be unbounded from above.
Moreover, assume
\begin{equation}
\label{cond:1}
\exists \, l > 0 : \int_{t_0}^t \frac{(c-b)(\tau)}{1 + b(\tau)\mu(\tau)}\, \Delta\tau\, 
\leq l, \quad \text{for all}\,\, t \geq t_0,
\end{equation}
and 
\begin{equation}
\label{cond:2}
\int_{t_0}^\infty b(\tau)\, \Delta\tau = \infty.
\end{equation}
Under these conditions, all solutions of \eqref{dynamic:sir} converge to the 
equilibrium $(0,0,N)$ with $N = x_0 + y_0 + z_0$.
\end{theorem}

\begin{proof}
Let $f(t)$ and $g(t)$ be as defined in Theorem \ref{thm:6}. Following 
Lemma \ref{lemma:1}, we have
$$
0 < e_{\left(\frac{c-b}{1+b\mu}\right)}(t,t_0) 
\leq \exp\left\{\int_{t_0}^t \frac{(c-b)(\tau)}{1+b(\tau)\mu(\tau)}\,\Delta\tau\right\}
\displaystyle \overset{\eqref{cond:1}}{\leq} e^l, \quad \text{for all}\,\,t \geq t_0. 
$$
Since $g(t) \geq 0$ (Theorem \ref{thm:6} and Lemma \ref{lemma:3}) one can apply
Lemma \ref{lemma:2} obtaining
\begin{equation*}
\begin{aligned}
e_g(t,t_0) &\geq 1 + \int_{t_0}^t g(\tau)\,\Delta\tau = 1 
+ \int_{t_0}^t \frac{b(\tau)\kappa}{e_{\frac{c-b}{1+b\mu}}(t,t_0) + \kappa}\, \Delta\tau\\
& \geq 1 + \frac{\kappa}{e^l + \kappa}\int_{t_0}^t b(\tau)\,\Delta\tau
\overset{\eqref{cond:2}}{\rightarrow} \infty.
\end{aligned}
\end{equation*}
So, as $t \rightarrow \infty$, $e_g(t,t_0) \rightarrow \infty$, which means that 
$e_{\ominus g}(t,t_0) \rightarrow 0$. Thus, from Theorem \ref{thm:6},
$$
\lim_ {t \rightarrow \infty} x(t) = \lim_ {t \rightarrow \infty} y(t) = 0,
\quad \lim_ {t \rightarrow \infty} z(t) = N.
$$
This ends the proof.
\end{proof}

\begin{corollary}
If $b(t) = c(t)$ for all $t \in \mathbb{T}$, then Theorem \ref{thm:7} holds 
provided
$$
\int_{t_0}^\infty b(\tau)\,\Delta\tau = \infty.
$$
\end{corollary}

The following corollary is a direct consequence of Theorem~\ref{thm:6}.

\begin{corollary}[Explicit solution to the discrete-time SIR model]	
Let $\mathbb{T} = h\mathbb{Z}$, $h > 0$. 
In this case, system \eqref{dynamic:sir} becomes
\begin{equation}
\label{discrete:hz}
\begin{cases}
\dfrac{x(t+1) - x(t)}{h} = -\dfrac{b(t)x(t+1)y(t)}{x(t) + y(t)},\\
\dfrac{y(t+1) - y(t)}{h} = \dfrac{b(t)x(t+1)y(t)}{x(t) + y(t)} - c(t)y(t+1),\\
\dfrac{z(t+1) - z(t)}{h} = c(t)y(t+1),
\end{cases}
\end{equation}
where $b, c : \mathbb{T} \rightarrow \mathbb{R}_0^+$, $x(t_0) = x_0 > 0,\, 
y(t_0) = y_0 > 0\,\, \text{and} \,\, z(t_0) = z_0 \geq 0.$
The unique solution of system \eqref{discrete:hz} is given by
\begin{equation}
\label{solution:discrete}
\begin{cases}
x(t_n) = x_0 \displaystyle \prod_{i = 0}^{n-1} \left(\frac{\prod_{j=0}^{i-1} 
\xi(t_j) + \kappa}{\prod_{j=0}^{i-1} \xi(t_j) + \kappa + b(t_i)\kappa\,h}
\right),\\
y(t_n) = \frac{y_0}{\prod_{i = 0}^{n-1} \xi(t_i)} \displaystyle 
\prod_{i = 0}^{n-1} \left(\frac{\prod_{j=0}^{i-1} \xi(t_j) + \kappa}
{\prod_{j=0}^{i-1} \xi(t_j) + \kappa + b(t_i)\kappa\,h}\right),\\
z(t_n) = N - \left(x_0 + \frac{y_0}{\prod_{i = 0}^{n-1} \xi(t_i)}\right)
\displaystyle \prod_{i = 0}^{n-1} \left(\frac{\prod_{j=0}^{i-1} \xi(t_j) + \kappa}
{\prod_{j=0}^{i-1} \xi(t_j) + \kappa + b(t_i)\kappa\,h}\right),
\end{cases}
\end{equation}
where $\kappa = \frac{y_0}{x_0}$, $N = x_0 + y_0 + z_0$ and 
$\xi(t) = \frac{1 + c(t)h}{1 + b(t)h}$.
\end{corollary}

\begin{remark}
If $b,\, c \in \mathbb{R}^+_0$, then solution \eqref{solution:discrete} 
coincides with the one first obtained in \cite{sir:discrete}.
\end{remark}

\begin{example}
\label{ex:1}
To illustrate Theorem~\ref{thm:7}, consider a disease with a periodic transmission 
rate, for example, due to seasonal factors that influence how people come into 
contact with the disease or how susceptible they are. In this scenario, 
let us consider 
$$
b(t) = 0.8 + 0.6\sin(mt), \quad m \in \mathbb{R}\backslash\{0\}.
$$ 
Considering medical advances over time, we let $c(t) = \frac{t}{1+t}$. 
For $\mathbb{T} = h\mathbb{Z}$, $h > 0$, 
the dynamics of system \eqref{solution:discrete} is as given in 
Fig.~\ref{zero:convergence}: the solution converges to the extinction
of susceptible and infected individuals.
\end{example}

\begin{figure}[ht]
\centering
\includegraphics[scale=0.7]{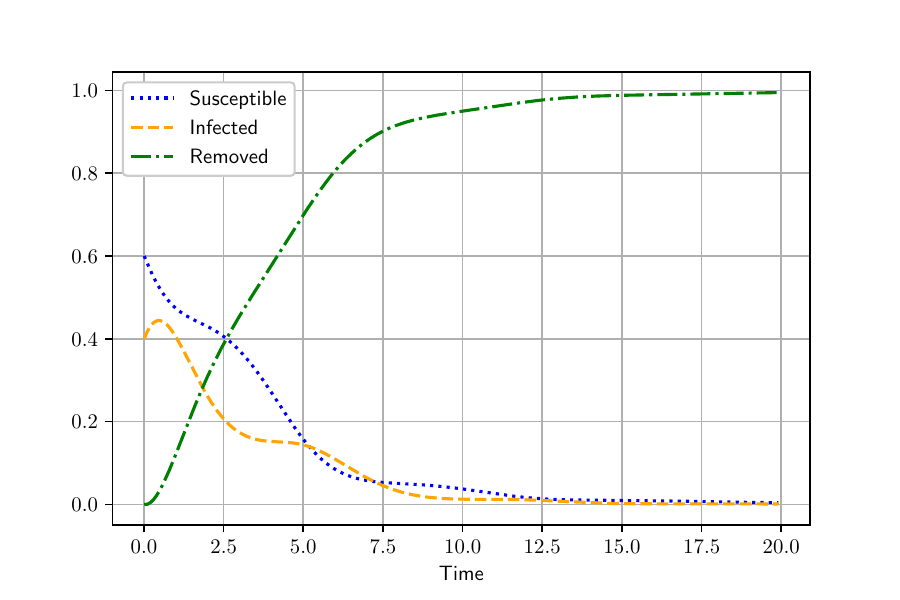}
\caption{Dynamics of Example \ref{ex:1} with $x_0 = 0.6$, $y_0 = 0.4$ and $m = -1$.}
\label{zero:convergence} 
\end{figure}

\begin{theorem}[Asymptotic stability of the disease-free equilibrium]
\label{thm:8}
Consider the SIR model \eqref{dynamic:sir} 
and let $\mathbb{T}$ be unbounded from above. Assume 
\begin{equation}
\label{cond:3}
\exists\, m > 0 : b(t) \leq \frac{m(c-b)(t)}{1 + b(t)\mu(t)},
\quad \text{for all}\,\, t \in \mathbb{T},
\end{equation}
and 
\begin{equation}
\label{cond:4}
\int_{t_0}^\infty \frac{(c-b)(t)}{1 + b(t)\mu(t)} \Delta\tau = \infty.
\end{equation}
Thus, all solutions of \eqref{dynamic:sir} converge to the equilibrium
$(\alpha, 0, N-\alpha)$, where $0<~\alpha~\leq~N$.  
\end{theorem}

\begin{proof}
Let $f(t)$ and $g(t)$ be as defined in Theorem~\ref{thm:6}
by \eqref{eq:f:g}. First, note that 
\eqref{cond:3} implies that $c(t) \geq b(t)$, for all $t \in \mathbb{T}$, 
and thus $f(t) \geq 0$. Following Lemma \ref{lemma:2}, we get
$$
e_f(t,t_0) \geq 1 + \int_{t_0}^t f(\tau) \Delta\tau 
\overset{\eqref{cond:4}}{\rightarrow} \infty.
$$
This means that
$\displaystyle \lim_{t\rightarrow \infty} e_{\ominus f}(t,t_0) = 0$. 
Now note that 
$$
g(t) = \frac{b(t)\kappa}{e_{f}(t,t_0) + \kappa}
\leq \frac{b(t)\kappa}{e_{f}(t,t_0)},
$$
and by condition \eqref{cond:3}, 
$$
g(t) \leq m\kappa f(t) \cdot 
\frac{1}{e_{f}(t,t_0)}.
$$
Thus, 
\begin{equation*}
\begin{aligned}
\int_{t_0}^t g(\tau) \Delta\tau 
&\leq m\kappa \int_{t_0}^{t} f(\tau)
\cdot \frac{1}{e_{f}(\tau,t_0)}\,\Delta\tau\\
&= m\kappa \left[\frac{1}{e_{f}(t,t_0)} - 1\right]
\leq m\kappa.
\end{aligned}
\end{equation*} 
Next, since $g(t) \geq 0$ (Theorem \ref{thm:6} and Lemma \ref{lemma:3}), it 
follows from Lemma \ref{lemma:2} that
$$
1 \leq 1 + \int_{t_0}^t g(\tau)\,\Delta\tau \leq e_g(t,t_0)
\leq \exp\left\{\int_{t_0}^t g(\tau)\,\Delta\tau\right\} \leq e^{m\kappa}.
$$ 
Therefore, $\lim_{t \rightarrow \infty} e_g(t,t_0)$ exists and it is bounded 
from below by 1 and from above by $e^{m\kappa}$. It is now straightforward to 
note that 
$$
\lim_{t \rightarrow \infty} e_{\ominus g}(t,t_0) = e^{-m\kappa} > 0.
$$
Thus, from Theorem~\ref{thm:6}, it follows that
$$
\lim_{t \rightarrow \infty} x(t) = \alpha > 0,\quad 
\lim_{t \rightarrow \infty} y(t) = 0,\quad
\lim_{t \rightarrow \infty} z(t) = N - \alpha,
$$
ending our proof. 
\end{proof}

An important time scale in physics is the quantum time scale
$\mathbb{T} = q^{\mathbb{N}_0}$ with $q > 1$ \cite{quantum:calculus}.
In this case, the delta-derivative of Definition~\ref{delta:derv}
reduces to the quantum derivative  
$$
D_qf(x) = \frac{f(qx) - f(x)}{(q-1)t}.
$$
The points of the quantum time-scale are denoted by
\begin{equation}
\label{remark:2}
t_0, \quad t_1 = qt_0, \quad t_2 = qt_1 = q^2t_0,\quad \ldots, \quad t_n = q^nt_0,
\end{equation}
where $t_0$ is some given positive real number.

As a corollary of Theorem~\ref{thm:6}, we obtain the exact solution of 
the quantum SIR model. 

\begin{corollary}[Explicit solution to the quantum SIR model]
\label{quantum:exact}
Let  $\mathbb{T} = q^{\mathbb{N}_0}$ and $q > 1$. Then,
\eqref{dynamic:sir} reduces to the form
\begin{equation}
\label{quantum}
\begin{cases}
\dfrac{x(qt) - x(t)}{(q-1)t} = -\dfrac{b(t)x(qt)y(t)}{x(t) + y´(t)},\\
\dfrac{y(qt) - y(t)}{(q-1)t} = \dfrac{b(t)x(qt)y(t)}{x(t) + y(t)} - c(t)y(qt),\\
\dfrac{z(qt) - z(t)}{(q-1)t} = c(t)y(qt),
\end{cases}
\end{equation}
where $b, c : \mathbb{T} \rightarrow \mathbb{R}_0^+$, $x(t_0),\, y(t_0) > 0$ 
and $z(t_0) \geq 0$.
The unique solution of system \eqref{quantum} is given by 
\begin{equation}
\label{solution:quantum}
\begin{cases}
x(t_n) = x(t_0) \displaystyle\prod_{s=0}^{n-1} 
\frac{\prod_{i=0}^{s-1} \xi(t_i) + \kappa}
{\prod_{i=0}^{s-1} \xi(t_i) + (q-1)b(s)s\kappa  + \kappa},\\
y(t_n) = \frac{y(t_0)}{\prod_{s=0}^{n-1} \xi(t_s)} 
\displaystyle\prod_{s=0}^{n-1} 
\frac{\prod_{i=0}^{s-1} \xi(t_i) + \kappa}
{\prod_{i=0}^{s-1} \xi(t_i) + (q-1)b(s)s\kappa  + \kappa},\\
z(t) = N - \left(x(t_0) + \frac{y(t_0)}{\prod_{s=0}^{n-1} \xi(t_s)}\right)
\displaystyle\prod_{s=0}^{n-1} 
\frac{\prod_{i=0}^{s-1} \xi(t_i) + \kappa}
{\prod_{i=0}^{s-1} \xi(t_i) + (q-1)b(s)s\kappa  + \kappa},
\end{cases}
\end{equation}
where $\kappa = \frac{y_0}{x_0}$, $N = x_0 + y_0 + z_0$ and 
$\xi(t) = \frac{1 + c(t)(q-1)t}{1 + b(t)(q-1)t}$ and $t$ being as defined
in \eqref{remark:2}.
\end{corollary}

\begin{example}
\label{ex:2}
To illustrate Theorem~\ref{thm:8}, we opt for the probability density function 
of the log-normal distribution and for a ``von Bertalanffy'' type function to 
model the time-varying parameters $b$ and $c$, respectively. By doing so, we 
are modeling a transmission rate that grows very rapidly at the onset of the 
epidemic before declining, a pattern commonly observed due to early ignorance 
followed by increasing precautions taken by the susceptible population. On the 
other hand, the use of a ``von Bertalanffy'' function suggests that we are
modeling a removal rate that begins at a very low level, improves rapidly due 
to adaptations, and eventually approaches a natural upper limit. More precisely,
let 
$$
b(t) = \frac{1}{t\sqrt{2\pi}\sigma} e^{-\frac{(\ln(t) - \mu)^2}{2\sigma^2}},\quad
c(t) = \gamma_\infty(1 - e^{-kt - a}).
$$
Specifically, we choose $\mathbb{T} = q^{\mathbb{N}_0}$ with $q > 1$,
$\sigma = 0.7$ $\mu = 1$, $\gamma_\infty = 0.3$,
$k = 0.5$ and $a = 0.3$. The corresponding dynamics of system \eqref{quantum} 
is illustrated in Fig.~\ref{alpha:convergence}, showing convergence
to the disease-free equilibrium.
\end{example}

\begin{figure}[ht]
\centering
\includegraphics[scale=0.7]{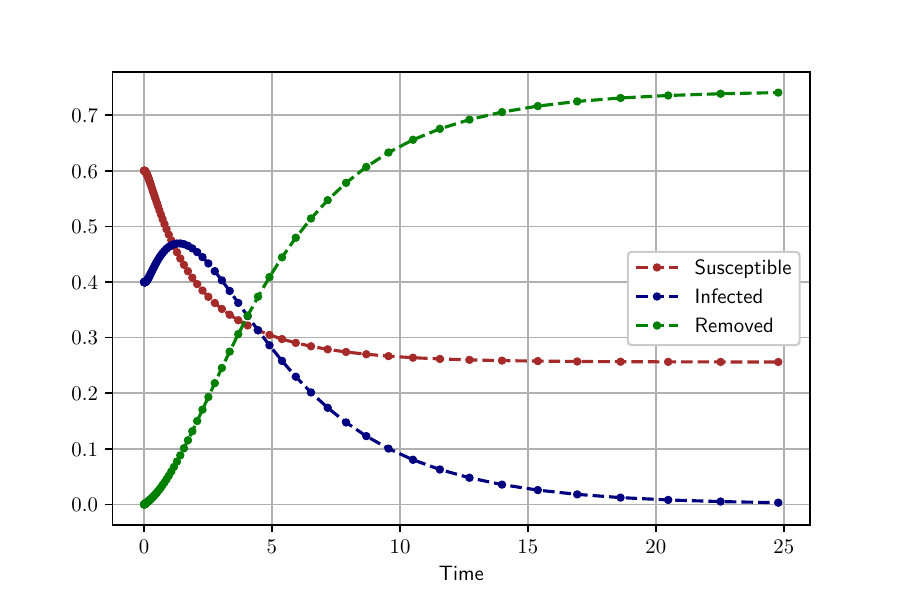}
\caption{Dynamics of Example \ref{ex:2} with $x_0 = 0.6$, $y_0 = 0.4$ and $q = 1.1$.}
\label{alpha:convergence} 
\end{figure}

The reproduction number $\mathcal{R}_0$, defined as the expected number of 
secondary cases produced by a single infection in a fully susceptible population, 
is one of the most important thresholds in mathematical epidemiology. It is 
well-known that for system \eqref{sir:continuous} with positive constant rates one has
\begin{equation}
\label{eq:brn}
\mathcal{R}_0 = \frac{b}{c}
\end{equation}
(see, e.g., \cite{R0}). This quantity remains
invariant across any arbitrary time-scale, as it is determined solely by the 
dynamics of the disease, independent of time. 

\begin{corollary}
If $b,\,c \in \mathbb{R}^+_0$, then the conclusion of Theorem \ref{thm:7} remains 
valid provided $\mathcal{R}_0 \geq 1$. 
\end{corollary}

\begin{corollary}
If $b,\,c \in \mathbb{R}^+_0$, then the conclusion of Theorem \ref{thm:8} holds
provided $\mathcal{R}_0 < 1$.
\end{corollary}

In the following result, we 
establish a relationship between the infection
and removed rates and the behavior of the 
infected population. 

\begin{theorem}[Necessary conditions for the monotonic behavior of the infected population]
\label{thm:9}
For all $t \in \mathbb{T}$, the following statements are true:
\begin{itemize}
\item if $\frac{b(t)}{c(t)} > \frac{x(t)+y(t)}{x^{\sigma}(t)}$, then $y^{\Delta} > 0$;
\item if $\frac{b(t)}{c(t)} < \frac{x(t)+y(t)}{x^{\sigma}(t)}$, then $y^{\Delta} < 0$;
\item if $\frac{b(t)}{c(t)} = \frac{x(t)+y(t)}{x^{\sigma}(t)}$, then the infected 
population remains constant.    
\end{itemize}	
\end{theorem}

\begin{proof}
Consider the dynamic SIR model \eqref{dynamic:sir}. 
Since $y$ is differentiable at $t$, then 
$y^{\sigma}(t) = y(t) + \mu(t)y^{\Delta}(t)$. Thus, the second equation of 
\eqref{dynamic:sir} becomes 
\begin{equation*}
\begin{aligned}
y^{\Delta} &= \frac{b(t)x^{\sigma}(t)y(t)}{x(t) + y(t)} - c(t)[y(t) + \mu(t)y^\Delta(t)]\\
\Leftrightarrow y^{\Delta}(t) &= \frac{y(t)}{1+c(t)\mu(t)}
\left(\frac{b(t)x^{\sigma}(t)}{x(t)+y(t)} - c(t)\right).
\end{aligned}
\end{equation*}
Since $y(t),\,c(t),\,\mu(t)$ are positive, the behavior of $y$ depends solely on 
the value of 
$$
\frac{b(t)x^{\sigma}(t)}{x(t)+y(t)} - c(t).
$$
Clearly, if $\frac{b(t)x^{\sigma}(t)}{x(t)+y(t)} > c(t)$, then $y^{\Delta} > 0$. This
is equivalent to 
$$
\frac{b(t)}{c(t)} > \frac{x(t)+y(t)}{x^{\sigma}(t)}.
$$
On the other hand, if 
$$
\frac{b(t)x^{\sigma}(t)}{x(t)+y(t)} < c(t),
$$ 
then $y^{\Delta} < 0$. Finally, if 
$$
\frac{b(t)x^{\sigma}(t)}{x(t)+y(t)} = c(t),
$$
then $y^{\Delta} = 0$, which corresponds to a constant population. 
The result is proved. 
\end{proof}

We now obtain the exact solution of system \eqref{dynamic:sir} when 
$\mathbb{T} = \mathbb{R}$. This result is a direct consequence of 
Theorem~\ref{thm:6} and was first obtained in \cite{sir:ts}.

\begin{corollary}[Explicit solution to the continuous SIR model with time-dependent 
infection and removed rates]
The unique solution of the continuous SIR system 
\begin{equation}
\label{sir:continuous2}
\begin{cases}
x'(t) = -\dfrac{b(t) x(t) y(t)}{x(t)+y(t)}, \\
y'(t) = \dfrac{b(t) x(t) y(t)}{x(t)+y(t)} - c(t)y(t), \\
z'(t) = c(t)y(t),
\end{cases}
\end{equation}
is given by
\begin{equation}
\label{exact:cont}
\begin{cases}
\displaystyle
x(t) = x_0 \exp\left\{-\kappa \int_{t_0}^t b(s) 
\left(e^{\int_{t_0}^{s}(c-b)(\tau)\,d\tau} + \kappa\right)^{-1}\,ds\right\},\\
\displaystyle
y(t) = y_0 \exp\left\{\int_{t_0}^{t}\left[b(s)\left(1 + \kappa
e^{\int_{t_0}^{s}(b-c)(\tau)\,d\tau}\right)^{-1} - c(s)\right]\,ds\right\},\\
\displaystyle
z(t) = N - \left(y_0e^{\int_{t_0}^t (b-c)(s)\,ds} + x_0\right)
\exp\left\{-\kappa\int_{t_0}^{t} b(s)\left(\kappa + 
e^{\int_{t_0}^s (c-b)(\tau)\,d\tau}\right)^{-1}\,ds\right\},
\end{cases}
\end{equation}
where $N = x_0 + y_0 + z_0$ and $\kappa = \frac{y_0}{x_0}$.
\end{corollary}

\begin{remark}
For constants $b$ and $c$, \eqref{exact:cont} reduces to the classical
solution presented in \cite{cont:sol}. 
\end{remark}

\begin{example}
\label{ex:03}
To illustrate Theorem~\ref{thm:9}, consider $\mathbb{T} = \mathbb{R}$ and 
$b,\,c \in \mathbb{R}^+_0$. Fig.~\ref{infected} presents the solution of $y(t)$ 
while Fig.~\ref{r0} compares the value of the basic reproduction number \eqref{eq:brn}
$\mathcal{R}_0 = 1.5$ with $\frac{x(t)+y(t)}{x^{\sigma}(t)}$, in accordance
with Theorem~\ref{thm:9}. 
\end{example}

\begin{figure}[ht]
\centering
\begin{subfigure}[b]{0.7\textwidth}
\includegraphics[width=1\linewidth]{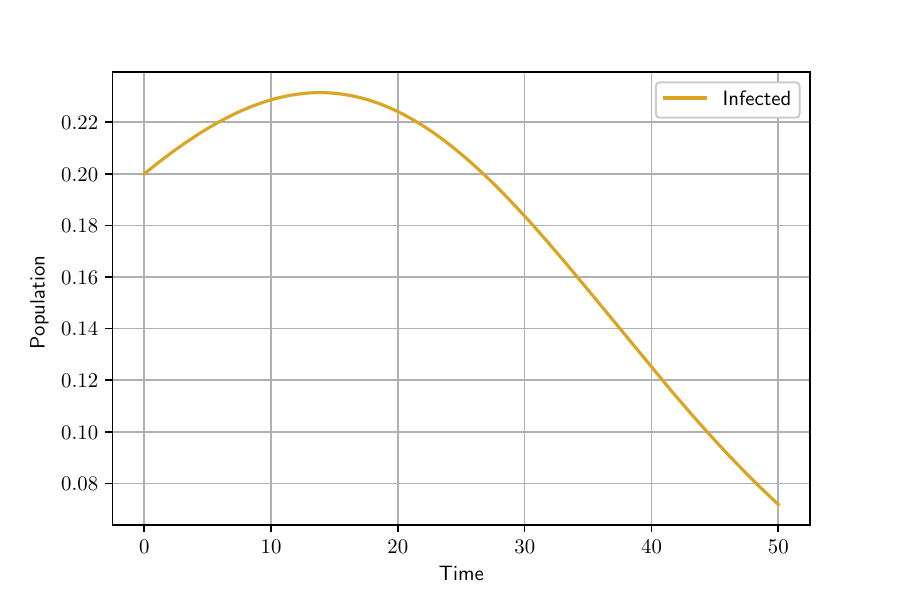}
\caption{Infected population}
\label{infected} 
\end{subfigure}
\begin{subfigure}[b]{0.7\textwidth}
\includegraphics[width=1\linewidth]{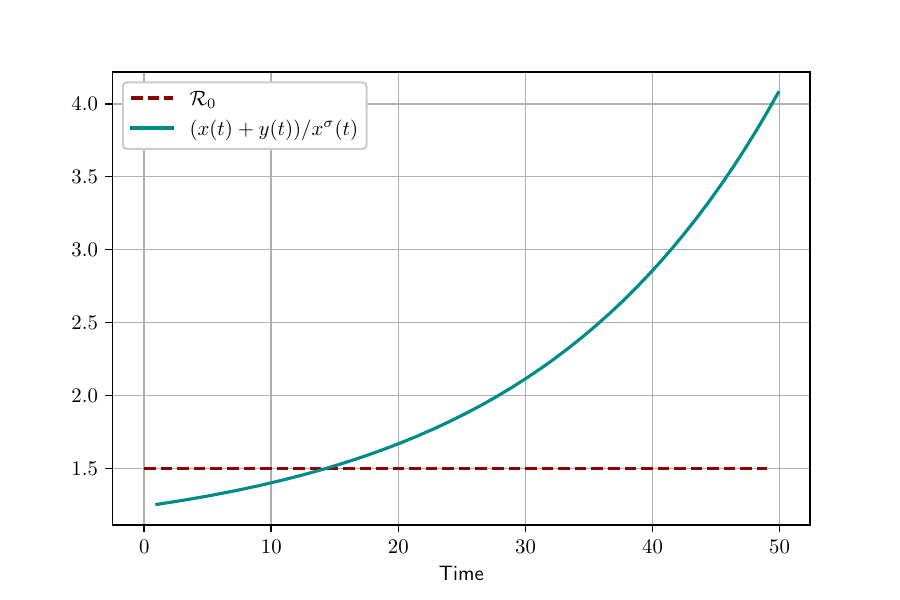}
\caption{$\mathcal{R}_0 = 1.5$ and $\frac{x(t)+y(t)}{x^{\sigma}(t)}$}
\label{r0}
\end{subfigure}
\caption{Illustration of Theorem~\ref{thm:9}: 
dynamics of the infected population $y(t)$ of Example~\ref{ex:03}
with $x_0 = 0.8$, $y_0 = 0.2$, $b = 0.15$ and $c = 0.1$.}
\end{figure}

%------------------------------------------------------

\section{Conclusion}
\label{sec:conc}

The SIR model is the most classical and fundamental model for understanding disease 
dynamics. It has been extensively studied and applied to describe the spread of infectious 
diseases in populations, see e.g. \cite{MR4368314,s11071-024-09710-9,s11071-025-11006-5}.
Our novelty here lies in extending the classical SIR model to the theory of time scales in a 
consistent way. 

We introduced a new dynamic epidemic model on arbitrary time scales, 
based on classical Bailey's SIR model, and derive its exact solution. 
In contrast with available results in the literature, our model has
always a non-negative solution, which is a crucial aspect
from the applications point of view. We analyze the asymptotic behavior of the 
susceptible, infected and removed individuals, proving model's 
consistency for any arbitrary time scale. Along the manuscript, 
we present examples of applications in discrete, quantum, 
and continuous time domains. The new model proposed in our work guarantees the
non-negativity of solutions, in contrast with \cite{sir:ts}, which does not. 
Our results add biological relevance to time-scale models.

The use of time scales provides a unifying mathematical framework that 
integrates discrete, continuous, and more general types of time domains. This flexibility is 
especially valuable in biological and epidemiological modeling, where data and processes often 
occur at irregular or mixed time intervals. Thus, time scales allow for more accurate and 
versatile modeling of biological systems by capturing dynamics that are not purely continuous 
or discrete, but may combine both aspects.
The main advantage of our model is precisely this adaptability: by formulating the SIR 
dynamics on arbitrary time scales and understanding its behavior, we can better represent 
real-world scenarios where disease transmission and population changes happen at non-uniform 
time scales. This strongly enhances SIR-type model's applicability and relevance.

Much remains to be done. For example, the inclusion of parameters, 
such as natality and mortality, increases the complexity 
of the analysis of the system and, to the best of our knowledge, 
finding an exact solution for such models remains an open problem. 
Here we considered the basic SIR framework, which already captures 
essential epidemic dynamics and allows for analytical 
tractability. For future work one may 
address more realistic models, e.g. 
with natality and mortality parameters. 
This will allow for a more accurate and comprehensive 
understanding of disease dynamics and potential control strategies.
Another interesting line of research consists to generalize our results
for fractional systems on time scales \cite{Hattaf2024,Torres,TM2}.

%------------------------------------------------------

\backmatter

\bmhead{Author contributions}

Márcia Lemos-Silva: Conceptualization, Formal analysis, Investigation, Methodology, 
Software, Validation, Visualization, Writing -- original draft, Writing -- review \& editing. 
Sandra Vaz: Conceptualization, Formal analysis, Investigation, Methodology,
Supervision, Validation, Writing -- original draft, Writing -- review \& editing. 
Delﬁm F.M. Torres: Conceptualization, Formal analysis, Investigation, Methodology, 
Project administration, Supervision, Validation, Writing -- original draft, Writing -- review \& editing.

\bmhead{Funding}

Lemos-Silva and Torres are supported by CIDMA 
and Vaz by CMA-UBI under 
FCT (\url{https://ror.org/00snfqn58})  
Multi-Annual Financing Program for R\&D Units.
Lemos-Silva is also supported by the FCT PhD grant
UI/BD/154853/2023 (\url{https://doi.org/10.54499/UI/BD/154853/2023})
and by an ERASMUS+ fellowship; 
Torres via project CoSysM3, 
Reference 2022.03091.PTDC 
(\url{https://doi.org/10.54499/2022.03091.PTDC}).

\bmhead{Availability of data and materials} 

No datasets were generated or analyzed during the current study.

\subsection*{Declarations}

\bmhead{Conflict of interest}

The authors declare no competing interests.

%------------------------------------------------------

%------------------------------------------------------

\end{document}